\documentclass[aps,prx,twocolumn,showpacs,amsmath,notitlepage,amssymb,superscriptaddress]{revtex4-1}
\usepackage{bm}        % for math
\usepackage{amssymb}   % for math
\usepackage{amsmath}
\usepackage{amsthm}
\usepackage{graphicx}  % needed for figures
\usepackage{color}

\definecolor{dartmouthgreen}{rgb}{0.05, 0.5, 0.06}

\newcommand{\bra}[1]{\langle {#1} |}
\newcommand{\ket}[1]{| {#1} \rangle}

\newtheorem{lemma}{Lemma}

\newtheorem{theorem}{Theorem}

\begin{document}

\title{Bipartite discrimination of independently prepared quantum states as a counterexample to a parallel repetition conjecture}
\author{Seiseki Akibue}
\email{seiseki.akibue.rb@hco.ntt.co.jp}
 \affiliation{NTT Communication Science Laboratories, NTT Corporation 3-1 Morinosato Wakamiya, Atsugi-shi, Kanagawa 243-0124, JAPAN}
 
\author{Go Kato}
\email{kato.go@lab.ntt.co.jp}
 \affiliation{NTT Communication Science Laboratories, NTT Corporation 3-1 Morinosato Wakamiya, Atsugi-shi, Kanagawa 243-0124, JAPAN}

\date{\today}

\begin{abstract}
For distinguishing quantum states sampled from a fixed ensemble, the gap in bipartite and single-party distinguishability can be interpreted as a  {\it nonlocality of the ensemble.} In this paper, we consider bipartite state discrimination in a composite system consisting of $N$ subsystems, where each subsystem is shared between two parties and the state of each subsystem is randomly sampled from a particular ensemble comprising the Bell states. We show that the success probability of perfectly identifying the state converges to $1$ as $N\rightarrow\infty$ if the entropy of the probability distribution associated with the ensemble is less than $1$, even if the success probability is less than $1$ for any finite $N$. In other words, the nonlocality of the $N$-{\it fold} ensemble asymptotically disappears if the probability distribution associated with each ensemble is concentrated. Furthermore, we show that the disappearance of the nonlocality can be regarded as a remarkable counterexample of a fundamental open question in theoretical computer science, called a {\it parallel repetition conjecture} of {\it interactive games} with two classically communicating players. Measurements for the discrimination task include a projective measurement of one party represented by stabilizer states, which enable the other party to perfectly distinguish states that are sampled with high probability.
\end{abstract}

\maketitle

\section{Introduction}
Various aspects of nonlocal properties of quantum mechanics have been investigated by considering multipartite information-processing tasks undertaken by joint quantum operations called {\it local operations and classical communication} (LOCC). Indeed, considered not to increase quantum correlation between the parties, LOCC is widely used for characterizing entanglement measures \cite{VVedral, Horodecki, MBPlenio} and nonlocal properties of unitary operations \cite{Soeda1, Stahlke, Soeda2}.

Another aspect of nonlocal properties is characterized by considering bipartite state discrimination. Bipartite state discrimination is a task where two parties, typically called Alice and Bob, perform a measurement implemented by LOCC to distinguish states sampled from an a priori known fixed ensemble of quantum states. By definition, the ability of the parties in bipartite state discrimination is more restricted than in single-party state discrimination as illustrated in Fig. \ref{fig:intro}(a) and (b). However, if each state constituting the ensemble is a {\it classical state}, namely a probabilistic mixture of the tensor products of two fixed mutually orthogonal states, there is no gap in bipartite and single-party distinguishability. Thus, when the gap exists, it can be interpreted as a {\it nonlocality of the ensemble}, which has been extensively studied.

\begin{figure}
 \centering
  \includegraphics[height=.15\textheight]{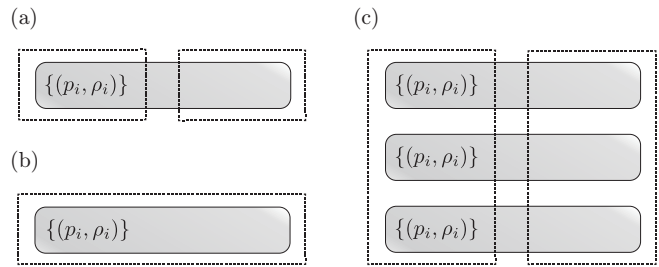}
  \caption{Graphical representations of three types of quantum state discrimination. Rounded rectangles represent quantum systems and dotted rectangles represent subsystems where joint quantum operations can be performed. The state of each system is randomly sampled from an a priori known ensemble $\{(p_i,\rho_i)\}$. (a) In bipartite state discrimination, a measurement is implemented by LOCC. (b) In single-party state discrimination, a measurement is implemented by a joint quantum operation. (c) Bipartite discrimination of quantum states sampled from a three-fold ensemble.}
\label{fig:intro}
\end{figure} 

Several studies have revealed the difference between the nonlocality of an ensemble and entanglement: for distinguishing any two entangled pure states, the ensemble is local, i.e., they can be optimally distinguished by LOCC as well as by joint measurement \cite{twostatediscrimination1, twostatediscrimination2}; there exist nonlocal ensembles comprising only product states \cite{productstate1, productstate2,productstate3,productstate4,productstate5}; and increasing the number of entangled states constituting an ensemble can decrease the nonlocality \cite{localdiscrimination}.
On the other hand, for distinguishing $M$ orthogonal maximally entangled states having local dimension $d$, any ensemble comprising such states (each state is sampled with non-zero probability) is nonlocal if $M>d$ \cite{optimalprob1, optimalprob2, optimalprob3, optimalprob4}. When $M=d$, any ensemble with $M=d=3$ is local \cite{optimalprob2} whereas there exists a nonlocal ensemble with $M=d=4$, which also demonstrates a novel phenomenon of entanglement discrimination catalysis \cite{optimalprob5}.
Furthermore, a novel application of the nonlocality of an ensemble was found in quantum data hiding \cite{datahiding1,datahiding2}. Recently, connections to other fundamental issues, such as the monogamy of entanglement \cite{datahidingstate2}, an area law \cite{datahidingstate3}, and a characterization of quantum mechanics in general probabilistic theories \cite{GPTs}, have also been found.

In this paper, we reveal a counterintuitive behavior of the nonlocality of an ensemble caused by entanglement, and show that it provides a remarkable counterexample of a fundamental open question in theoretical computer science, called a {\it parallel repetition conjecture} of {\it interactive games} with two players, which has been proven to be true in a classical scenario. We consider bipartite state discrimination in a composite system consisting of $N$ subsystems, where each subsystem is shared between Alice and Bob and the state of each subsystem is randomly sampled from a particular ensemble. In the discrimination task, they perform an LOCC measurement to distinguish states sampled from the $N$-{\it fold} ensemble formed by taking $N$ copies of one particular ensemble as illustrated in Fig. \ref{fig:intro}(c).

We consider that each ensemble comprises the Bell states, and show that the bipartite distinguishability approaches the single-party distinguishability as $N$ grows if the entropy of the probability distribution associated with each ensemble is less than a certain value. More precisely, we measure the distinguishability by the success probability of perfectly identifying a state used in the minimum-error discrimination \cite{statediscrimination}. We show that if the entropy condition is satisfied, the success probability in bipartite discrimination converges to $1$ as $N\rightarrow\infty$, even if the success probability in bipartite discrimination is less than $1$ for any finite $N$, namely,  the nonlocality of the $N$-fold ensembles asymptotically disappears. Since such a disappearance of the nonlocality occurs only if the mixed state corresponding to each ensemble is entangled, it also demonstrates a difference between the nonlocality of an ensemble and entanglement.

An {\it interactive proof system} is a fundamental notion of (probabilistic) computation in computational complexity theory \cite{AM,MIP,QIP,QMIP}, with important applications to modern cryptography \cite{ZK,QSZK} and hardness of approximation \cite{PCP1, PCP2}. Its general description is based on an {\it interactive game} \cite{quantumproofs} involving an interaction between a referee and players. The referee makes a fixed probabilistic trial to judge players to {\it win} or {\it lose} the game, and the players try to maximize the winning probability. If the maximum winning probability of an interactive game is less than $1$, it is natural to expect a {\it parallel repetition conjecture} of the game holds, namely, the maximum winning probability of the repeated game, where the referee simultaneously repeats the game independently and judges the players to win the repeated game if the players win all the games, decreases exponentially. If the parallel repetition conjecture holds, efficient error reduction of the computation in interactive proof systems is possible without increasing the round of interactions. The conjecture has been proven for interactive games with a single player \cite{KW, productrule3, quantumproofs} and with two separated classical players \cite{Raz, Holenstein}; however, it remains widely open as to whether the conjecture holds for interactive games with two quantum players, with several positive results for special cases \cite{parallelrep1, parallelrep2}.

Bipartite state discrimination can be regarded as an interactive game with two classically communicating quantum players, where the referee prepares the state of a composite system randomly sampled from an ensemble and judges the players to win the game if they guess the state correctly. To the best of our knowledge, it is unknown whether the parallel repetition conjecture of interactive games with two classically communicating players holds, which originates from an open problem posed in \cite{QMIPLOCC}. We show that the disappearance of the nonlocality in the state discrimination can be regarded as a remarkable counterexample of the conjecture, i.e., while the maximum winning probability of each game is less than $1$, that of the repeated game does not decrease; moreover, it asymptotically approaches $1$.

This paper is organized as follows: In Section II, we provide precise definitions of an $N$-fold ensemble and the distinguishability of states sampled from it and introduce some notations concerning the definitions. In Section III, we review some known upper bounds for the success probability of the identification in bipartite discrimination and apply them to show an asymptotic behavior of an upper bound of the success probability in our scenario. In Section IV, we construct an LOCC measurement for the success probability in bipartite discrimination to converge to $1$. In Section V, we review an interactive game and its parallel repetition conjecture and show that the disappearance of the nonlocality can be regarded as a counterexample of the parallel repetition conjecture of interactive games with two classically communicating players. The last section is devoted to conclusion and a discussion.

\section{Definitions and notations}
We denote the Hilbert space of Alice's system and Bob's system by $\mathcal{A}$ and $\mathcal{B}$, respectively. Suppose the state of the composite system $\mathcal{A}\otimes\mathcal{B}$ is randomly sampled from an a priori known ensemble of finite quantum states,
\begin{equation}
 \{(q_{\mathbf{m}},\ket{\Psi_{\mathbf{m}}}\in\mathcal{A}\otimes\mathcal{B})\}_{\mathbf{m}},
\end{equation}
where $q_{\mathbf{m}}$ is an element of a probability vector. (Note that in general, an ensemble can comprise mixed states in state discrimination; however, it is sufficient to consider pure states in our scenario.)

Moreover, we consider that Alice's system and Bob's system consist of $N$ subsystems $\mathcal{A}=\otimes_{n=1}^N\mathcal{A}_n$ and $\mathcal{B}=\otimes_{n=1}^N\mathcal{B}_n$ respectively, where $\mathcal{A}_n=\mathcal{B}_n=\mathbb{C}^2$ for all $n$, and a state of each subsystem $\mathcal{A}_n\otimes\mathcal{B}_n$ is randomly sampled from a particular ensemble comprising the Bell states $\{(p_{m},\ket{\Phi_{m}}\in\mathcal{A}_n\otimes\mathcal{B}_n):m\in F_2^2\}$, where $F_2^k$ is the direct product of $k$ finite fields of two elements, $\mathbf{p}=(p_{00},p_{01},p_{10},p_{11})$ is a probability vector, and
\begin{eqnarray}
 \ket{\Phi_{00}}=\frac{1}{\sqrt{2}}(\ket{00}+\ket{11}),\\
 \ket{\Phi_m}=(\mathbb{I}\otimes\sigma_{m})\ket{\Phi_{00}},\\
 \sigma_{00}=\mathbb{I},\sigma_{01}=X,\sigma_{10}=Z,\sigma_{11}=Y.
\end{eqnarray}
Note that $\mathbb{I}$ represents the identity operator, $X,Y,Z$ represent Pauli operators, and $\{\ket{0},\ket{1}\}$ is a fixed orthonormal basis of $\mathbb{C}^2$ such that $X\ket{x}=\ket{1-x}$, $Z\ket{x}=(-1)^x\ket{x}$ and $Y\ket{x}=(-1)^xi\ket{1-x}$ for $x\in \{0,1\}$. 

The $N$-fold ensemble formed by taking $N$ copies of an ensemble $\{(p_{m},\ket{\Phi_{m}}\in\mathcal{A}_n\otimes\mathcal{B}_n):m\in F_2^2\}$ is represented by $\{(q_{\mathbf{m}},\ket{\Psi_\mathbf{m}}\in\mathcal{A}\otimes\mathcal{B}):\mathbf{m}=(m_1,\cdots,m_N)\in F_2^{2N}\}$ such that
\begin{eqnarray}
q_{\mathbf{m}}&=&\prod_{n=1}^N p_{m_n}\label{eq:productprob}\\
 \ket{\Psi_{\mathbf{m}}}&=&\otimes_{n=1}^N\ket{\Phi_{m_n}}\\
 &=&(\mathbb{I}^{(\mathcal{A})}\otimes\sigma_{\mathbf{m}}^{(\mathcal{B})})\ket{\Phi_{00}}^{\otimes N},
\end{eqnarray}
where $\sigma_{\mathbf{m}}^{(\mathcal{B})}=\otimes_{n=1}^N\sigma_{m_n}$, the superscript of a linear operator represents the Hilbert space it acts on, and the order of the Hilbert spaces is appropriately permuted in $\otimes_{n=1}^N\ket{\Phi_{m_n}}$ and $\ket{\Phi_{00}}^{\otimes N}$.

Alice and Bob's measurement can be described by a {\it positive-operator valued measure} (POVM) $\{M_{\hat{\mathbf{m}}}\in P(\mathcal{A}\otimes\mathcal{B})\}_{\hat{\mathbf{m}}}$ satisfying $\sum_{\hat{\mathbf{m}}} M_{\hat{\mathbf{m}}}=\mathbb{I}$, where $P(\mathcal{A}\otimes\mathcal{B})$ represents a set of positive semidefinite operators on $\mathcal{A}\otimes\mathcal{B}$. When a state $\ket{\Psi_{\mathbf{m}}}$ is sampled, a measurement outcome $\hat{\mathbf{m}}$, corresponding to their estimation of $\mathbf{m}$, is obtained with probability given by $\bra{\Psi_{\mathbf{m}}}M_{\hat{\mathbf{m}}}\ket{\Psi_{\mathbf{m}}}$. Thus, the success probability of perfectly identifying a state is given by $\sum_{\mathbf{m}} q_{\mathbf{m}}\bra{\Psi_{\mathbf{m}}}M_{\mathbf{m}}\ket{\Psi_{\mathbf{m}}}$. The bipartite distinguishability of states sampled from an $N$-fold ensemble is measured by the {\it maximum} success probability of the identification, namely,
\begin{equation}
\label{eq:LOCCprob}
 \gamma=\sup\left\{\sum_{\mathbf{m}\in F_2^{2N}} q_{\mathbf{m}}\bra{\Psi_{\mathbf{m}}}M_{\mathbf{m}}\ket{\Psi_{\mathbf{m}}}:\{M_{\hat{\mathbf{m}}}\}_{\hat{\mathbf{m}}}\in  LOCC\right\},
\end{equation}
where $LOCC$ represents a set of POVMs implemented by LOCC between Alice and Bob. Note that the single-party distinguishability, where the supremum is taken over all POVMs in Eq.\eqref{eq:LOCCprob}, is always $1$ for any $N$ and $\mathbf{p}$ since $\{\ket{\Psi_{\mathbf{m}}}\}_{\mathbf{m}\in F_2^{2N}}$ is a set of orthogonal states.
As other measures of the distinguishability, the maximum success probability of unambiguous state discrimination \cite{unambiguous1, unambiguous2, unambiguous3} and the separable fidelity \cite{sepfid} have been also studied.

In our construction of an LOCC measurement given in Section \ref{sec:construction}, it is sufficient to consider an important subset of $LOCC$, a set of POVMs implemented by {\it one-way} LOCC from Alice to Bob.
Indeed, in many cases, one-way LOCC is sufficient for perfect discrimination when perfect bipartite state discrimination is possible \cite{twostatediscrimination1, localdiscrimination, optimalprob2, optimalprob5}.
In one-way LOCC (from Alice to Bob), first Alice performs a measurement on her own system described by a POVM $\{A_{\mathbf{a}}\in P(\mathcal{A})\}_{\mathbf{a}}$ and sends the measurement outcome $\mathbf{a}$ to Bob. Then Bob performs a measurement on his own system described by a POVM $\{B_{\hat{\mathbf{m}}|\mathbf{a}}\in P(\mathcal{B})\}_{\hat{\mathbf{m}}}$ based on $\mathbf{a}$. Thus, the maximum success probability of the identification by one-way LOCC is given by
\begin{eqnarray}
  \gamma_1=\max\Big\{\sum_{\mathbf{m}\in F_2^{2N}}\sum_{\mathbf{a}} q_{\mathbf{m}}\bra{\Psi_{\mathbf{m}}}A_{\mathbf{a}}\otimes B_{\mathbf{m}|\mathbf{a}}\ket{\Psi_{\mathbf{m}}}\nonumber\\
 :\{A_{\mathbf{a}}\}_{\mathbf{a}}{\rm \ and\ }\{B_{\hat{\mathbf{m}}|\mathbf{a}}\}_{\hat{\mathbf{m}}}{\rm \ are\ POVMs}\Big\}.
\label{eq:1LOCCprob}
\end{eqnarray}
By definition, $\gamma\geq\gamma_1$. Note that $\gamma_1$ is always achievable by some measurements implemented by one-way LOCC due to its compactness, in contrast to general LOCC \cite{productstate4, CLMOW}.

Since the parameters in an $N$-fold ensemble consist only of a probability vector $\mathbf{p}=(p_{00},p_{01},p_{10},p_{11})$ and the number of subsystems $N$, in the following sections, we denote the maximum success probabilities of the identification defined in Eq. \eqref{eq:LOCCprob} and Eq. \eqref{eq:1LOCCprob} by $\gamma^{(N)}(\mathbf{p})$ and $\gamma_1^{(N)}(\mathbf{p})$, respectively.

\section{Upper bound of LOCC measurements}
\label{sec:upperbound}
In \cite{optimalprob2} (simpler proof in \cite{optimalprob4}), it was shown that the success probability of identifying a state randomly sampled from an ensemble comprising $M$ equiprobable maximally entangled states having local dimension $d$ is at most $d/M$. Since $\ket{\Psi_{\mathbf{m}}}$ is a maximally entangled state having local dimension $2^N$ for any $\mathbf{m}\in F_2^{2N}$, by applying the result, we obtain
\begin{equation}
 \gamma^{(N)}\left(\left(\frac{1}{4},\frac{1}{4},\frac{1}{4},\frac{1}{4}\right)\right)\leq\frac{1}{2^N}.
 \label{eq:Nanthan}
\end{equation}
The upper bound in the right-hand side is achievable by performing one-way LOCC measurements to each subsystem independently: for each subsystem $\mathcal{A}_n\otimes\mathcal{B}_n$, Alice and Bob measure their own subsystem with respect to the fixed basis $\{\ket{0},\ket{1}\}$, compare the measurement results by one-way classical communication from Alice to Bob, and Bob guesses $m_n$ as $00$ if the measurement results agree and $m_n$ as $01$ if they disagree. 

By applying Theorem 4 in \cite{optimalprob6}, an upper bound of $\gamma^{(N)}(\mathbf{p})$ for a non-uniform probability vector $\mathbf{p}$ is obtained:
\begin{equation}
\label{eq:upperbound}
 \gamma^{(N)}(\mathbf{p})\leq\max\left\{\sum_{\mathbf{m}\in X}q_{\mathbf{m}}:X\subseteq F_2^{2N},|X|=2^N\right\}.
\end{equation}
Note that this upper bound can also be obtained by simply using Eq. \eqref{eq:Nanthan} as shown in Appendix \ref{appendix:upperbound}. For $N=1$, the upper bound is tight, namely,
\begin{equation}
 \gamma^{(1)}(\mathbf{p})= p_{a_0}+p_{a_1},
\end{equation}
where $\{a_0,a_1,a_2,a_3\}=F_2^2$ such that $p_{a_0}\geq p_{a_1}\geq p_{a_2}\geq p_{a_3}$. Since any set of two Bell states is locally unitarily equivalent to $\{\ket{\Phi_{00}},\ket{\Phi_{01}}\}$ \cite{LUequivalent}, the upper bound is achievable by one-way LOCC.

Using Eq. \eqref{eq:upperbound}, we can easily verify that the success probability $\gamma^{(N)}(\mathbf{p})$ is less than $1$ for any finite $N$ if and only if the number of non-zero elements in a probability vector $\mathbf{p}$ is greater than or equal to $3$. Furthermore, we can show a condition where the success probability $\gamma^{(N)}(\mathbf{p})$ converges to 0 as $N\rightarrow\infty$.

\begin{theorem}
\label{theorem:upperbound}
 Let $H(\mathbf{p})=-\sum_{x\in F_2^{2}}p_x\log p_x$ be the entropy of a probability vector $\mathbf{p}=(p_{00},p_{01},p_{10},p_{11})$. If $H(\mathbf{p})>1$,
\begin{equation}
 \lim_{N\rightarrow\infty}\gamma^{(N)}(\mathbf{p})=0.
\end{equation}
\end{theorem}

\begin{proof}
 We define a set of {\it typical sequences} $T(\epsilon)\subseteq F_2^{2N}$ as
 \begin{equation}
 T(\epsilon)=\left\{\mathbf{m}\in F_2^{2N}:\left|-\frac{1}{N}\log q_\mathbf{m}-H(\mathbf{p})\right|<\epsilon\right\}
\end{equation}
for $\epsilon>0$. It is obvious that if $\mathbf{m}\in T(\epsilon)$,
\begin{equation}
\label{eq:rangeqm}
 2^{-N(H(\mathbf{p})+\epsilon)}<q_{\mathbf{m}}< 2^{-N(H(\mathbf{p})-\epsilon)},
\end{equation}
and thus
\begin{equation}
 |T(\epsilon)|<2^{N(H(\mathbf{p})+\epsilon)}.
\end{equation}
By the {\it asymptotic equipartition property},
\begin{equation}
\label{eq:AEP}
  \sum_{\mathbf{m}\notin T(\epsilon)}q_{\mathbf{m}}\leq 2\exp\left(-2\frac{\epsilon^2}{\Delta^2}N\right),
\end{equation}
where $\Delta$ is a non-negative real number defined by $\Delta=\log (\max\{p_x\})-\log (\min\{p_x\})$.
%(Numerical calculations indicate $V$ is bounded, $V< 2.2$, for any probability vector $\mathbf{p}$.) 
An explicit derivation of Eq. \eqref{eq:AEP} is given in Appendix \ref{appendix:AEP}. Therefore, for any $N$ and $\epsilon>0$ and for any $X\subseteq F_2^{2N}$ satisfying $|X|=2^N$,
\begin{eqnarray}
 \sum_{\mathbf{m}\in X}q_{\mathbf{m}}&\leq&\sum_{\mathbf{m}\in X\cap T(\epsilon)}q_{\mathbf{m}}+\sum_{\mathbf{m}\notin T(\epsilon)}q_{\mathbf{m}}\\
 \label{eq:upperb}
 &<&2^{N(1-H(\mathbf{p})+\epsilon)}+2\exp\left(-2\frac{\epsilon^2}{\Delta^2}N\right).
\end{eqnarray}
Hence, if $1-H(\mathbf{p})<0$, there exists $\epsilon>0$ such that the right-hand side converges to $0$ as $N\rightarrow\infty$ since $\Delta$ is a constant when $N$ changes. Since the right hand side of Eq. \eqref{eq:upperbound} is also bounded by Eq. \eqref{eq:upperb}, this completes the proof.
\end{proof}

Note that this convergence condition is tight in the sense that there exists a probability vector $\mathbf{p}$ such that $H(\mathbf{p})\geq1$ but the success probability $\gamma^{(N)}(\mathbf{p})$ does not converge to $0$. Indeed, for $\mathbf{p}=\left(\frac{1}{2},\frac{1}{2},0,0\right)$, $H(\mathbf{p})=1$ and $\gamma^{(N)}(\mathbf{p})=1$ for any $N$.

\section{Construction of an LOCC measurement}
\label{sec:construction}
In this section, we show that the success probability of the identification, $\gamma^{(N)}(\mathbf{p})$, converges to $1$ if the entropy of a probability vector $H(\mathbf{p})$ is less than $1$ by constructing a one-way LOCC measurement.
The one-way LOCC measurement consists two steps:
\begin{enumerate}
 \item Alice performs a projective measurement described by $\{A_{\mathbf{a}}=\ket{\phi_{\mathbf{a}}}\bra{\phi_{\mathbf{a}}}\}_{\mathbf{a}}$, where $\{\ket{\phi_{\mathbf{a}}}\in\mathcal{A}\}_{\mathbf{a}}$ is an orthonormal basis of $\mathcal{A}$. 
 
 \item Bob performs a measurement on his system depending on Alice's measurement outcome described by $\{B_{\hat{\mathbf{m}}|\mathbf{a}}:\hat{\mathbf{m}}\in F_2^{2N}\}$. 
 \end{enumerate}

When a state $\ket{\Psi_{\mathbf{m}}}$ is sampled, the (unnormalized) state of Bob's system after Alice's measurement is given by
\begin{eqnarray}
 (\bra{\phi_{\mathbf{a}}}\otimes\mathbb{I}^{(\mathcal{B})})\ket{\Psi_{\mathbf{m}}}&=&(\bra{\phi_{\mathbf{a}}}\otimes\sigma_{\mathbf{m}}^{(\mathcal{B})})\ket{\Phi_{00}}^{\otimes N}\\
 &=&\frac{1}{\sqrt{2^N}}\sigma_{\mathbf{m}}^{(\mathcal{B})}\ket{\phi^*_{\mathbf{a}}},
\end{eqnarray}
where $\ket{\phi^*_{\mathbf{a}}}\in\mathcal{B}$ is the complex conjugate of $\ket{\phi_{\mathbf{a}}}$ with respect to the fixed basis. Note that $\{\ket{\phi^*_{\mathbf{a}}}\in\mathcal{B}\}_{\mathbf{a}}$ is an orthonormal basis if and only if $\{\ket{\phi_{\mathbf{a}}}\in\mathcal{A}\}_{\mathbf{a}}$ is. Therefore, the success probability $\gamma_1^{(N)}(\mathbf{p})$ defined by Eq. \eqref{eq:1LOCCprob} is bounded by
\begin{eqnarray}
\label{eq:lowerbound}
 \gamma_1^{(N)}(\mathbf{p})\geq\frac{1}{2^N}\sum_{\mathbf{m}\in F_2^{2N}}\sum_{\mathbf{a}}q_{\mathbf{m}}\bra{\phi^*_{\mathbf{a}}}\sigma_{\mathbf{m}}B_{\mathbf{m}|\mathbf{a}}\sigma_{\mathbf{m}}\ket{\phi^*_{\mathbf{a}}}.
\end{eqnarray}

We choose each state $\ket{\phi^*_{\mathbf{a}}}\in\mathcal{B}$ from a {\it stabilizer state}, which is widely used in quantum error correction \cite{GotPhD}, quantum computation \cite{GotKnill}, and measurement-based quantum computation \cite{MBQC}. Suppose $\mathcal{S}$ is a subgroup of an $N$-qubit Pauli group $\{\pm1,\pm i\}\times\{\sigma_{\mathbf{s}}:\mathbf{s}\in F_2^{2N}\}$. An $N$-qubit state $\ket{\psi}\in\mathbb{C}^{2^N}$ is {\it stabilized} by $\mathcal{S}$ if $\ket{\psi}$ is a simultaneous eigenstate of all elements of $\mathcal{S}$ with the eigenvalue $+1$:
\begin{equation}
 \forall S\in\mathcal{S},S\ket{\psi}=\ket{\psi}.
\end{equation}
It is known that stabilized state $\ket{\psi}$ is uniquely determined (up to a global phase) if and only if subgroup $\mathcal{S}$ is generated as a product of generators $\langle g_1,\cdots,g_N\rangle$, where each generator $g_n$ is taken from a subset of the Pauli group as $g_n\in\{\pm \sigma_{\mathbf{s}(n)}\}$, and the generators are commutative and {\it independent} in the sense that $\{\mathbf{s}(n)\in F_2^{2N}\}_{n=1}^N$ is linearly independent. Note that an orthonormal basis of $N$-qubit $\{\ket{\psi_{\mathbf{a}}}:\mathbf{a}=(a_1,\cdots,a_N)\in F_2^{N}\}$ can be constructed by taking each $\ket{\psi_{\mathbf{a}}}$ as a state stabilized by $\langle (-1)^{a_1}g_1,\cdots,(-1)^{a_N}g_N\rangle$ since two eigenspaces of the Pauli group corresponding to different eigenvalues are orthogonal. If we construct an orthonormal basis $\{\ket{\phi^*_{\mathbf{a}}}\in\mathcal{B}\}_{\mathbf{a}}$ using stabilizer states, Bob's measurement can be significantly simplified using the following lemma:

\begin{lemma}
\label{lemma:unitary}
 Let $\{\ket{\psi_{\mathbf{a}}}\in\mathbb{C}^{2^N}:\mathbf{a}=(a_1,\cdots,a_N)\in F_2^N\}$ be an orthonormal basis stabilized by
   \begin{equation}
 \langle(-1)^{a_1}g_1,\cdots,(-1)^{a_N}g_N\rangle,
 \label{eq:stabilizer}
\end{equation}
 where $\{g_n\}_{n=1}^N$ is a set of commutative and independent elements of $\{\sigma_{\mathbf{s}}:\mathbf{s}\in F_2^{2N}\}$. Then, for any $\mathbf{a}\in F_2^{N}$, there exists a unitary operator $U_{\mathbf{a}}\in U(\mathbb{C}^{2^N})$ such that for any $\mathbf{m}\in F_2^{2N}$,
 \begin{equation}
 \sigma_{\mathbf{m}}\ket{\psi_{\mathbf{a}}}\propto U_{\mathbf{a}}\sigma_{\mathbf{m}}\ket{\psi_{\mathbf{0}}},
\end{equation}
where $U(\mathbb{C}^{2^N})$ represents a set of $N$-qubit unitary operators.
\end{lemma}

\begin{proof}
Let $g_n=\sigma_{\mathbf{s}(n)}$, where $\{\mathbf{s}(n)\in F_2^{2N}\}_{n=1}^N$ is linearly independent. Let $G=(\mathbf{s}(1),\cdots,\mathbf{s}(N))^T$ be a $N\times 2N$ matrix over $F_2$. By straightforward calculation, we obtain
 \begin{equation}
 \sigma_{\mathbf{m}}\ket{\psi_{\mathbf{a}}}\propto \ket{\psi_{\mathbf{a}+GP\mathbf{m}}},
\end{equation}
where $P$ is a $2N\times 2N$ matrix over $F_2$ such that
\begin{equation}
 P=\oplus_{n=1}^N
 \begin{pmatrix}
 0&&1\\
 1&&0
\end{pmatrix}.
\end{equation}
Since $rank(G)=rank(GP)=N$, there exists linearly independent $N$ columns in $GP$. Thus,
\begin{equation}
 \exists U_{\mathbf{a}}\in U(\mathbb{C}^{2^N}),\forall\mathbf{m}\in F_2^{2N},\ket{\psi_{\mathbf{a}+GP\mathbf{m}}}\propto U_{\mathbf{a}}\ket{\psi_{GP\mathbf{m}}}
\end{equation}
is equivalent to
\begin{equation}
 \exists U_{\mathbf{a}}\in U(\mathbb{C}^{2^N}),\forall\mathbf{c}\in F_2^{N},\ket{\psi_{\mathbf{c}+\mathbf{a}}}\propto U_{\mathbf{a}}\ket{\psi_{\mathbf{c}}}.
\end{equation}
This is true since $\{\ket{\psi_{\mathbf{c}+\mathbf{a}}}\}_{\mathbf{c}\in F_2^N}$ is an orthonormal basis for any $\mathbf{a}\in F_2^N$.
\end{proof}

Suppose $\{\ket{\phi^*_{\mathbf{a}}}\in\mathcal{B}\}_{\mathbf{a}\in F_2^N}$ is an orthonormal basis $\{\ket{\psi_{\mathbf{a}}}\}_{\mathbf{a}\in F_2^N}$ defined in Lemma \ref{lemma:unitary}, and Bob's measurement is represented by $B_{\hat{\mathbf{m}}|\mathbf{a}}=U_{\mathbf{a}}B_{\hat{\mathbf{m}}}U_{\mathbf{a}}^{\dag}$, where $\{B_{\hat{\mathbf{m}}}\in P(\mathcal{B})\}_{\hat{\mathbf{m}}\in F_2^{2N}}$ is a POVM and $\{U_{\mathbf{a}}\}_{\mathbf{a}\in F_2^N}$ is a set of unitary operators defined in Lemma \ref{lemma:unitary}. Due to Lemma \ref{lemma:unitary} and Eq. \eqref{eq:lowerbound}, the success probability $\gamma_1^{(N)}(\mathbf{p})$ is bounded by
\begin{eqnarray}
\label{eq:lowerbound2}
\gamma_1^{(N)}(\mathbf{p})&\geq& \eta^{(N)}(\mathbf{p},\ket{\xi^{(N)}})\\
\eta^{(N)}(\mathbf{p},\ket{\xi^{(N)}})&:=& \max\Big\{\sum_{\mathbf{m}\in F_2^{2N}}q_{\mathbf{m}}\bra{\xi^{(N)}}\sigma_{\mathbf{m}}B_{\mathbf{m}}\sigma_{\mathbf{m}}\ket{\xi^{(N)}}:\nonumber\\
&&\{B_{\hat{\mathbf{m}}}\}_{\hat{\mathbf{m}}\in F_2^{2N}}{\rm\ is\ a\ POVM}\Big\},
\end{eqnarray}
where $\ket{\xi^{(N)}}:=\ket{\phi^*_{\mathbf{0}}}$ is a $N$-qubit state stabilized by $\langle g_1,\cdots, g_N\rangle$. Note that Bob's measurement is optimal in the sense that the maximum of the right-hand side of Eq. \eqref{eq:lowerbound} over $\{B_{\hat{\mathbf{m}}|\mathbf{a}}\}_{\hat{\mathbf{m}}\in F_2^{2N}}$ and $\eta^{(N)}(\mathbf{p},\ket{\xi^{(N)}})$ are the same.

The probability $\eta^{(N)}(\mathbf{p},\ket{\xi^{(N)}})$ can be understood in the scenario of quantum error correction, i.e., Alice sends an $N$-qubits stabilizer state $\ket{\xi^{(N)}}$ to Bob via a noisy channel. In the noisy channel, an error described by a Pauli operator $\sigma_{m}$ occurs on each qubit with probability $p_m$ independently and identically, and Bob tries to detect what types of error occurred. The probability $\eta^{(N)}(\mathbf{p},\ket{\xi^{(N)}})$ is equal to the maximum success probability of the perfect error detection. The existence of quantum error correction code suggests that faithful error detection is possible if the probability of error is less than a certain value. In the following theorem, we show that the probability $\eta^{(N)}(\mathbf{p},\ket{\xi^{(N)}})$ converges to $1$ if the entropy of the probability distribution of error $H(\mathbf{p})$ is less than $1$.

\begin{theorem}
\label{theorem:revival}
 If the entropy satisfies $H(\mathbf{p})<1$, there exist a set of stabilizer states $\{\ket{\xi^{(N)}}\in\mathbb{C}^{2^N}\}_{N\in\mathbb{N}}$ such that $\lim_{N\rightarrow\infty} \eta^{(N)}(\mathbf{p},\ket{\xi^{(N)}})=1$.
\end{theorem}

\begin{proof}
We show the existence of the set of stabilizer states using the idea of the {\it random coding}.
For any subspace $C\subseteq F_2^{2N}$, the {\it symplectic dual subspace} is defined by 
\begin{equation}
 C^{\bot}=\{\mathbf{u}\in F_2^{2N}:\forall \mathbf{v}\in C,\mathbf{u}\odot\mathbf{v}=0\},
\end{equation}
where $\odot$ denotes the symplectic product: $\mathbf{u}\odot\mathbf{v}=\mathbf{u}^TP\mathbf{v}$. Note that $\dim C+\dim C^{\bot}=2N$. $N$-dimensional subspace $C$ is called {\it symplectic self-dual} if $C=C^{\bot}$, or equivalently,
\begin{equation}
 \forall \mathbf{u},\mathbf{v}\in C, \mathbf{u}\odot \mathbf{v}=0.
\end{equation}

Suppose $\ket{\xi^{(N)}}$ is stabilized by $\langle \sigma_{\mathbf{s}(1)},\cdots,\sigma_{\mathbf{s}(N)}\rangle$, where $\{\mathbf{s}(n)\}_{n=1}^N$ is a basis of $N$-dimensional symplectic self-dual subspace $C$. Since $[\sigma_{\mathbf{s}(m)},\sigma_{\mathbf{s}(n)}]=0$ if and only if $\mathbf{s}(m)\odot\mathbf{s}(n)=0$, $\{\sigma_{\mathbf{s}(n)}\}_{n=1}^N$ is commutative and independent; thus, $\ket{\xi^{(N)}}$ is well defined.
 
Since the state with an error $\mathbf{m}$ is given by
\begin{equation}
 \sigma_{\mathbf{m}}\ket{\xi^{(N)}}\propto\ket{\psi_{GP\mathbf{m}}},
\end{equation}
where $\ket{\psi_{\mathbf{a}}}$ is a state stabilized by $\langle (-1)^{a_1}\sigma_{\mathbf{s}(1)},\cdots,(-1)^{a_N}\sigma_{\mathbf{s}(N)}\rangle$, and $G=(\mathbf{s}(1),\cdots,\mathbf{s}(N))^T$ is a matrix over $F_2$ as defined in the proof of Lemma \ref{lemma:unitary}, and since two states corresponding to errors $\mathbf{m}$ and $\mathbf{m}'$ are distinguishable if and only if $GP\mathbf{m}\neq GP\mathbf{m}'$, the Bob's optimal measurement detecting error is described by $\{B_{\hat{\mathbf{m}}(\mathbf{a})}=\ket{\psi_{\mathbf{a}}}\bra{\psi_{\mathbf{a}}}\}_{\mathbf{a}\in F_2^{N}}$, where $\hat{\mathbf{m}}:F_2^N\rightarrow F_2^{2N}$ is defined by
\begin{equation}
 \hat{\mathbf{m}}(\mathbf{a})=\arg\max_{\mathbf{m}\in F_2^{2N}}\{q_{\mathbf{m}}:GP\mathbf{m}=\mathbf{a}\},
\end{equation}
and for $\hat{\mathbf{m}}'\notin range(\hat{\mathbf{m}})$, $B_{\hat{\mathbf{m}}'}=0$.

Then, the failure probability of the error detection is given by
\begin{eqnarray}
 &&1-\eta^{(N)}(\mathbf{p},\ket{\xi^{(N)}})\\
 &\leq&\sum_{\mathbf{m}\in F_2^{2N}}q_{\mathbf{m}}\mathbf{I}\Big[\exists \mathbf{m}'\in F_2^{2N},\nonumber\\
 &&\mathbf{m}\neq\mathbf{m}'\wedge q_{\mathbf{m}'}\geq q_{\mathbf{m}}\wedge GP(\mathbf{m}+\mathbf{m}')=\mathbf{0}\Big]\\
 &\leq&\sum_{\mathbf{m}\in T(\epsilon)}q_{\mathbf{m}}\sum_{\mathbf{m}'\neq\mathbf{m}}\mathbf{I}\left[q_{\mathbf{m}'}\geq q_{\mathbf{m}}\right]\nonumber\\
 &&\mathbf{I}\left[ GP(\mathbf{m}+\mathbf{m}')=\mathbf{0}\right]+2\exp\left(-2\frac{\epsilon^2}{\Delta^2}N\right),
\end{eqnarray}
where $\mathbf{I}[L]$ is the {\it indicator function}, defined by $\mathbf{I}[L]=1$ if $L$ is true and $\mathbf{I}[L]=0$ if $L$ is false, and $T(\epsilon)$ is a set of typical sequences defined in the proof of Theorem \ref{theorem:upperbound}. Note that we used Eq.\eqref{eq:AEP} to derive the second inequality.

We calculate the expectation value of the failure probability when $N$-dimensional subspace $C$ is randomly sampled from sample space $\Omega=\{C\subset F_2^{2N}:C=C^{\bot}\}$ with a uniform probability. For any $\mathbf{c}(\neq\mathbf{0})\in F_2^{2N}$,
\begin{equation}
\label{eq:symplectic}
 E\left[\mathbf{I}\left[ GP\mathbf{c}=\mathbf{0}\right]\right] 
 =E\left[\mathbf{I}\left[\mathbf{c}\in C\right]\right]
 =\frac{1}{2^{N}+1}<\frac{1}{2^N}.
\end{equation}
The last equation is obtained by calculating the number of symplectic self-dual subspaces as shown in Appendix \ref{appendix:symplectic}. Thus, we obtain
\begin{eqnarray}
&& E\left[1-\eta^{(N)}(\mathbf{p},\ket{\xi^{(N)}})\right]\nonumber\\
 &< &2^{-N}\sum_{\mathbf{m}\in T(\epsilon)}q_{\mathbf{m}}\sum_{\mathbf{m}'\neq\mathbf{m}}\mathbf{I}\left[q_{\mathbf{m}'}\geq q_{\mathbf{m}}\right]\nonumber\\&&\ \ \ \ \ \ \ \ \ \ \ \ \ \ \ \ \ \ \ \ \ \ \ \ \ \ \ \ \ \ \ + 2\exp\left(-2\frac{\epsilon^2}{\Delta^2}N\right)\\
 &<& 2^{N(H(\mathbf{p})+\epsilon-1)}\sum_{\mathbf{m}\in T(\epsilon)}q_{\mathbf{m}}+ 2\exp\left(-2\frac{\epsilon^2}{\Delta^2}N\right)\\
 &\leq& 2^{N\left(H(\mathbf{p})+\epsilon-1\right)}+2\exp\left(-2\frac{\epsilon^2}{\Delta^2}N\right).
\end{eqnarray}
Note that we used Eq.\eqref{eq:rangeqm} to derive the second inequality.
Since there exists subspace $C$ bounded by the right-hand side for any $N$, this completes the proof.
\end{proof}

With Eq. \eqref{eq:lowerbound2}, this theorem implies the success probabilities of the identification, $\gamma_1^{(N)}(\mathbf{p})$ and $\gamma^{(N)}(\mathbf{p})$, converge to $1$ as $N\rightarrow\infty$ if $H(\mathbf{p})<1$. Note that if $H(\mathbf{p})<1$, the {\it quantum state merging} is possible without entanglement \cite{merging}; therefore, we can also construct a one-way LOCC measurement for the success probability $\gamma_1^{(N)}(\mathbf{p})$ to converge to $1$ by using the merging protocol proposed in \cite{merging}. However, our one-way LOCC measurement is easier to implement in the sense that the measurement is sampled from a finite set. Moreover, our measurement shows a closed connection between bipartite state discrimination and error correction.

Using the {\it positive partial transpose} (PPT) criterion \cite{PPT}, the mixed state corresponding to each ensemble $\rho=\sum_{m\in F_2^2}p_m\ket{\Phi_m}\bra{\Phi_m}$ is separable if and only if all the elements of probability vector $\mathbf{p}$ is less than or equal to $1/2$. We summarize properties of mixed state $\rho$ and success probability $\gamma^{(N)}(\mathbf{p})$ in Fig. \ref{fig:result} when probability vector $\mathbf{p}$ is characterized by two parameters, $s$ and $t$, as $\mathbf{p}=(s,t,1-s-t,0)$. As shown in the figure, there exists a region of $\mathbf{p}$, the interior of the white region of Fig. \ref{fig:result} (b), where the success probability of bipartite discrimination, $\gamma^{(N)}(\mathbf{p})$, is less than $1$ for any finite $N$ but converges to that of single-party discrimination as $N\rightarrow \infty$, i.e., the nonlocality of the $N$-fold ensembles asymptotically disappears. Note that mixed state $\rho$ is entangled in the region. On the other hand, if mixed state $\rho$ is separable and success probability $\gamma^{(N)}(\mathbf{p})$ is less than $1$ for some finite $N$, $\gamma^{(N)}(\mathbf{p})$ converges to $0$ as $N\rightarrow \infty$ as shown in Appendix \ref{appendix:mixed}.

\begin{figure}
 \centering
  \includegraphics[height=.16\textheight]{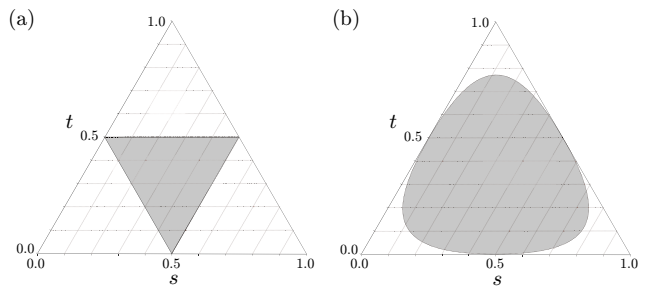}
  \caption{Properties of mixed state $\rho$ and success probability $\gamma^{(N)}(\mathbf{p})$ with probability vector $\mathbf{p}=(s,t,1-s-t,0)$. Probability vector $\mathbf{p}$ is represented by a point in a triangle or on its boundary in oblique coordinates. (a) Mixed state $\rho$ is separable if $\mathbf{p}$ is in the interior of the gray region or on its boundary, and entangled if $\mathbf{p}$ is in the exterior of the gray region. (b) Success probability $\gamma^{(N)}(\mathbf{p})$ is $1$ for any $N$ if $\mathbf{p}$ is on the boundary of the triangle and less than $1$ for any finite $N$ if $\mathbf{p}$ is in the triangle. The interior of the gray region corresponds to $H(\mathbf{p})>1$, where $\lim_{N\rightarrow\infty}\gamma^{(N)}(\mathbf{p})=0$. The exterior of the gray region corresponds to $H(\mathbf{p})<1$, where $\lim_{N\rightarrow\infty}\gamma^{(N)}(\mathbf{p})=1$.}
\label{fig:result}
\end{figure} 

A similar result can be found in \cite{similar}, where $N$-partite discrimination of three states sampled from an ensemble of $N$-copies of three unknown states was investigated; however, we investigate bipartite discrimination of $4^N$ states sampled from an $N$-fold ensemble in this paper.

\section{Bipartite state discrimination as an interactive game}
In general, an interactive game can be formulated by {\it quantum combs} \cite{comb} or {\it quantum strategies} \cite{strategy}. However, for our purpose, it is enough to use a normal quantum circuit description to introduce a two-turn interactive game between a referee and two classically communicating players as shown in Fig. \ref{fig:game}.
\begin{figure}
 \centering
  \includegraphics[height=.14\textheight]{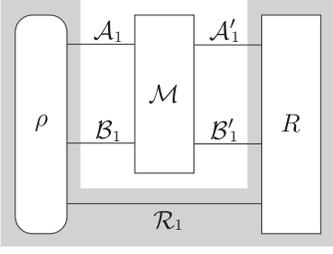}
  \caption{A two-turn interactive game between a referee and two classically communicating players, Alice and Bob. The two-turn interaction consists of quantum communication from the referee to the players and vice versa. The referee's operation (the shaded part) consists of preparing composite system $\mathcal{A}_1\otimes\mathcal{B}_1\otimes\mathcal{R}_1$ in mixed state $\rho$, sending subsystem $\mathcal{A}_1\otimes\mathcal{B}_1$ to the players, and performing a two-valued joint measurement on system $\mathcal{A}'_1\otimes\mathcal{B}'_1$ received from the players and his internal subsystem $\mathcal{R}_1$ to judge the players to win or lose the game. The players perform LOCC operation $\mathcal{M}$ to maximize the winning probability.}
\label{fig:game}
\end{figure} 
The two-turn interactive game consists the three steps:
\begin{enumerate}
 \item The referee prepares composite system $\mathcal{A}_1\otimes\mathcal{B}_1\otimes\mathcal{R}_1$ in mixed state $\rho \in D(\mathcal{A}_1\otimes\mathcal{B}_1\otimes\mathcal{R}_1)$, where $D(\mathcal{H}):=\{\rho\in P(\mathcal{H}):tr(\rho)=1\}$ represents a set of density operators on $\mathcal{H}$, and sends subsystems $\mathcal{A}_1$ and $\mathcal{B}_1$ to Alice and Bob, respectively.
 
 \item Alice and Bob perform LOCC operations described by linear map $\mathcal{M}:P(\mathcal{A}_1\otimes\mathcal{B}_1)\rightarrow P(\mathcal{A}'_1\otimes\mathcal{B}'_1)$ on subsystem $\mathcal{A}_1\otimes\mathcal{B}_1$ and send the resulting systems, $\mathcal{A}'_1\otimes\mathcal{B}'_1$, to the referee.
 
 \item The referee performs a measurement described by POVM $\{R,\mathbb{I}-R\}\subset P(\mathcal{A}'_1\otimes\mathcal{B}'_1\otimes\mathcal{R}_1)$ to judge Alice and Bob to win (corresponding to $R$) or lose (corresponding to $\mathbb{I}-R$) the game.
\end{enumerate}
Then, the maximum winning probability of the players is given by
\begin{equation}
 \chi^{(1)}=\sup\{tr\left(R\mathcal{M}\otimes\mathcal{I}(\rho)\right):\mathcal{M}{\rm\ is\ LOCC}\},
\end{equation}
where $\mathcal{M}\otimes\mathcal{I}:P(\mathcal{A}_1\otimes\mathcal{B}_1\otimes\mathcal{R}_1)\rightarrow P(\mathcal{A}'_1\otimes\mathcal{B}'_1\otimes\mathcal{R}_1)$ is a linear map satisfying $\mathcal{M}\otimes\mathcal{I}(V\otimes W)=\mathcal{M}(V)\otimes W$ for all $V\in P(\mathcal{A}_1\otimes\mathcal{B}_1)$ and $W\in P(\mathcal{R}_1)$. 

The repeated game is an interactive game where the referee simultaneously repeats one particular game independently and judges the players to win the repeated game if the players win all the games.
Thus, the $N$-times repeated game of the two-turn interactive game consists the three steps:
\begin{enumerate}
 \item The referee prepares $N$-copies of composite system $\mathcal{A}_1\otimes\mathcal{B}_1\otimes\mathcal{R}_1$ and sends subsystems $\otimes_{n=1}^N\mathcal{A}_n$ and $\otimes_{n=1}^N\mathcal{B}_n$ to Alice and Bob, respectively, where each composite system is labelled $\mathcal{A}_n\otimes\mathcal{B}_n\otimes\mathcal{R}_n$ $(n=1,\cdots,N)$.
 
 \item Alice and Bob perform LOCC operations described by linear map $\mathcal{M}:P((\otimes_{n=1}^N\mathcal{A}_n)\otimes(\otimes_{n=1}^N\mathcal{B}_n))\rightarrow P((\otimes_{n=1}^N\mathcal{A}_n')\otimes(\otimes_{n=1}^N\mathcal{B}_n'))$ and send the resulting systems to the referee, where each $\mathcal{A}_n'$ ($\mathcal{B}_n'$) has the same dimension as $\mathcal{A}'_1$ ($\mathcal{B}'_1$) in the single game.
 
 \item The referee performs a measurement described by POVM $\{R^{\otimes N},\mathbb{I}-R^{\otimes N}\}$ to judge Alice and Bob to win (corresponding to $R^{\otimes N}$) or lose (corresponding to $\mathbb{I}-R^{\otimes N}$) the game, where $\{R,\mathbb{I}-R\}\subset P(\mathcal{A}'_n\otimes\mathcal{B}'_n\otimes\mathcal{R}_n)$.
\end{enumerate}
Then, the maximum winning probability of the players is given by
\begin{equation}
  \chi^{(N)}=\sup\{tr\left(R^{\otimes N}\mathcal{M}\otimes\mathcal{I}(\rho^{\otimes N})\right):\mathcal{M}{\rm\ is\ LOCC}\},
\end{equation} 
where the order of the Hilbert spaces is appropriately permuted in $R^{\otimes N}$ and $\rho^{\otimes N}$.

A parallel repetition conjecture of an interactive game holds if $\chi^{(1)}<1$ implies $\chi^{(N)}<c^N$ with some constant $c<1$. It is easy to verify that the bipartite discrimination of states sampled from an ensemble comprising the Bell states task is a two-turn interactive game by setting
\begin{eqnarray}
 \rho&=&\sum_{m\in F_2^2}p_m\ket{\Phi_m}\bra{\Phi_m}^{(\mathcal{A}_n\otimes\mathcal{B}_n)}\otimes\ket{m}\bra{m}^{(\mathcal{R}_n)},\\
 R&=&\sum_{m\in F_2^2} \ket{m}\bra{m}^{(\mathcal{A}'_n)}\otimes \ket{m}\bra{m}^{(\mathcal{B}'_n)}\otimes \ket{m}\bra{m}^{(\mathcal{R}_n)},
\end{eqnarray}
where each subsystem $\mathcal{R}_n$ stores the label of a Bell state the referee sampled. The maximum winning probability of the players, $\chi^{(1)}$, is equal to the success probability of the identification, $\gamma^{(1)}(\mathbf{p})$. The maximum winning probability of the $N$-times repeated game, $\chi^{(N)}$, is equal to the success probability of the identification of an $N$-fold ensemble, $\gamma^{(N)}(\mathbf{p})$. Theorem \ref{theorem:revival} shows that there exist probability vectors $\mathbf{p}$ such that $\chi^{(N)}$ converges to $1$ as $N\rightarrow \infty$ while $\chi^{(1)}<1$, which is a remarkable counterexample of the parallel repetition conjecture.

\section{Conclusion and Discussion}
We have investigated bipartite discrimination of states sampled from an $N$-fold ensemble comprising the Bell states. We showed that the success probability of the perfect identification, $\gamma^{(N)}(\mathbf{p})$, converges to $1$ as $N\rightarrow\infty$ if the entropy of the probability distribution associated with each ensemble, $H(\mathbf{p})$, is less than $1$, even if $\gamma^{(N)}(\mathbf{p})<1$ for any finite $N$, namely, the nonlocality of the $N$-fold ensemble asymptotically disappears. Furthermore, the disappearance of the nonlocality can be regarded as a remarkable counterexample of the parallel repetition conjecture of interactive games with two classically communicating players.

Conversely, if $H(\mathbf{p})>1$, the quantum state merging is impossible without entanglement; moreover, we showed that $\gamma^{(N)}(\mathbf{p})$ converges to $0$ as $N\rightarrow\infty$. Therefore, our result also demonstrates a significant gap of the distinguishability with respect to $\gamma^{(N)}(\mathbf{p})$ between a mergeable ensemble and a non-mergeable ensemble.
Note that there does not always exist such a gap for ensembles comprising more general states, e.g., $N$-fold ensemble formed by ensemble $\{(s,\ket{01}),(t,\ket{10}),(1-s-t,\ket{\Phi_{00}})\}$ is perfectly distinguishable for any probability vector $(s,t,1-s-t)$ and any $N$ while the ensemble is not mergeable for particular probability vectors as shown in Fig. \ref{fig:merging}.
There remains a future work in investigating the gap of the distinguishability between a mergeable ensemble and a non-mergeable ensemble comprising more general states.

\begin{figure}
 \centering
  \includegraphics[height=.15\textheight]{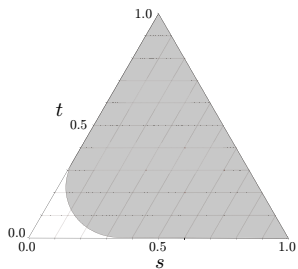}
  \caption{The possibility of quantum state merging for ensemble $\{(s,\ket{01}),(t,\ket{10}),(1-s-t,\ket{\Phi_{00}})\}$. In the gray region, the quantum state merging is impossible without entanglement, on the other hand, it is possible in the white region without entanglement}
\label{fig:merging}
\end{figure}

We can also discuss our result in another context. Intuitively, the optimal distinguishability in $N$ independent subsystems can be achieved by performing measurements on each subsystem independently. Indeed, in the case of single-party discrimination of quantum states sampled from independent (but not necessarily identical) ensembles, independent measurements can extract as much information about the composite system as any joint measurement \cite{productrule1} as depicted in Fig. \ref{fig:discussion} (a) and (b). The result was extended to the joint estimation of the parameters encoded in independent processes, where the optimal joint estimation can be achieved by estimating each process independently \cite{productrule2}.

However, our result shows that the optimal distinguishability of an $N$-fold ensemble cannot be achieved by independent LOCC measurement as depicted in Fig. \ref{fig:discussion} (c) but can be achieved by joint LOCC measurement, where Alice and Bob perform entangled measurement within their own system. 
%This can be understood intuitively as follows: when each ensemble possesses entanglement, Alice and Bob can use a part of the entanglement as an LOCC resource to perform joint quantum operations and improve the distinguishability compared to discrimination done by independent LOCC measurement. However, to achieve the {\it perfect} distinguishability, they use entanglement in the ensemble as an LOCC resource, and at the same time they have to distinguish entangled states constituting the ensemble. Indeed, there exists probability vector $\mathbf{p}$ such that the mixed state corresponding to each ensemble is entangled but $H(\mathbf{p})>1$, i.e., the distinguishability goes to zero as shown in Fig. \ref{fig:result}.

\begin{figure}
 \centering
  \includegraphics[height=.15\textheight]{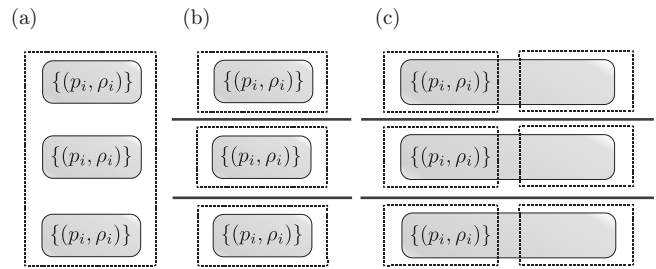}
  \caption{Graphical representations of three types of state discrimination where the state of each subsystem is randomly sampled from independent ensemble $\{(p_i,\rho_i)\}$. Rounded rectangles represent quantum systems and dotted rectangles represent subsystems where joint quantum operations can be performed. Communication across the bold line is forbidden. (a) Single-party state discrimination by joint measurement. (b) Single-party state discrimination by independent measurement. (c) Bipartite state discrimination by independent LOCC measurement.}
\label{fig:discussion}
\end{figure} 

\acknowledgements
We are greatly indebted to Seiichiro Tani, Kohtaro Suzuki, Hiroki Takesue, Mio Murao, Marco Tulio Quintino, Mateus Araujo, and Takuya Ikuta for their valuable discussions.

\appendix

\section{Upper bound of LOCC measurements}
\label{appendix:upperbound}
In this Appendix, we derive the upper bound in Eq. \eqref{eq:upperbound} by applying the following elementary lemma:

 \begin{lemma}
 \label{lemma:satprob}
 Let $\{a_k\in[0,1]\}_{k=1}^K$ be a set of real numbers, and $(\lambda_1,\cdots,\lambda_K)$ be a probability vector. If $\sum_{k=1}^K a_k\leq \tilde{K}$ for a non-negative integer $\tilde{K}\leq K$, then
 \begin{equation}
 \sum_{k=1}^K \lambda_ka_k\leq\max\left\{\sum_{k\in X}\lambda_k:X\subseteq [K],|X|= \tilde{K}\right\},
\end{equation}
where $[K]=\{1,2,\cdots,K\}$.
\end{lemma}

\begin{proof}
Suppose $X^*$ maximizes the right-hand side, and let $X^{*c}=[K]\setminus X^*$ be the complement of $X^*$. Let $\alpha=\sum_{k\in X^*}a_k$ and $\beta=\sum_{k\in X^{*c}}a_k$. Then $ \tilde{K}-\alpha\geq\beta\geq 0$. Let $\mu=\min\{\lambda_k:k\in X^*\}$ and $\nu=\max\{\lambda_k:k\in X^{*c}\}$. Then $\mu\geq \nu\geq 0$, and we obtain
 
\begin{eqnarray}
&& \sum_{k\in X^*}\lambda_k-\sum_{k=1}^K \lambda_ka_k\nonumber\\
&=&\sum_{k\in X^*}\lambda_k(1-a_k)-\sum_{k\in X^{*c}}\lambda_ka_k\\
 &\geq&\sum_{k\in X^*}\mu(1-a_k)-\sum_{k\in X^{*c}}\nu a_k\\
 &=&\tilde{K}\mu-\alpha \mu-\beta \nu\\
 &\geq&(\tilde{K}-\alpha)(\mu-\nu)\geq0.
\end{eqnarray}
\end{proof}

Eq.\eqref{eq:Nanthan} implies that for any LOCC measurement $\{M_{\hat{\mathbf{m}}}\}_{\hat{\mathbf{m}}}$,
\begin{equation}
 \sum_{\mathbf{m}\in F_2^{2N}}\bra{\Psi_{\mathbf{m}}}M_{\mathbf{m}}\ket{\Psi_{\mathbf{m}}}\leq 2^N.
\end{equation}
Since $\bra{\Psi_{\mathbf{m}}}M_{\mathbf{m}}\ket{\Psi_{\mathbf{m}}}\in [0,1]$, by applying Lemma \ref{lemma:satprob}, we obtain
\begin{eqnarray}
 &&\sum_{\mathbf{m}\in F_2^{2N}}q_{\mathbf{m}}\bra{\Psi_{\mathbf{m}}}M_{\mathbf{m}}\ket{\Psi_{\mathbf{m}}}\nonumber\\
 &\leq&\max\left\{\sum_{\mathbf{m}\in X}q_{\mathbf{m}}:X\subseteq F_2^{2N},|X|= 2^N\right\}
\end{eqnarray}
for any LOCC measurement $\{M_{\hat{\mathbf{m}}}\}_{\hat{\mathbf{m}}}$. Hence, the upper bound in Eq. \eqref{eq:upperbound} is derived.

\section{Asymptotic equipartition property}
\label{appendix:AEP}
In this Appendix, we derive Eq. \eqref{eq:AEP}.
Define a set of indices of the Bell states associated with non-zero probability $F=\{x\in F_2^2:p_x>0\}$. Let $\Omega=F^N$ be a sample space and $p(\mathbf{m})=q_{\mathbf{m}}$ be the probability mass function. Define random variables $Y_n(\mathbf{m})=-\log p_{m_n}$ and $Y(\mathbf{m})=\frac{1}{N}\sum_{n=1}^N Y_n(\mathbf{m})$, which are well-defined for $\mathbf{m}=(m_1,\cdots,m_N)\in\Omega$. Then $\{Y_n\}_{n=1}^N$ are mutually independent random variables, and
\begin{eqnarray}
 E[Y]=E[Y_n]=H(\mathbf{p}),\\
 \Delta:=\max\{\log p_{max}-\log p_{min},\delta\},
\end{eqnarray}
where $p_{max}=\max\{p_x:x\in F\}$, $p_{min}=\min\{p_x:x\in F\}$ and $\delta$ is an arbitrary positive real number. By Hoeffding's's inequality, for any $\epsilon>0$,
\begin{eqnarray}
 Pr[|Y-E[Y]|\geq\epsilon]\leq 2\exp\left(-2\frac{\epsilon^2}{\Delta^2}N\right).
\end{eqnarray}
Since
\begin{eqnarray}
 Pr[|Y-E[Y]|\geq\epsilon]&=&\sum_{\mathbf{m}\in \Omega\setminus T(\epsilon)}q_{\mathbf{m}}\\ &=&\sum_{\mathbf{m}\in F_2^{2N}\setminus T(\epsilon)}q_{\mathbf{m}},
\end{eqnarray}
Eq. \eqref{eq:AEP} is derived.

\section{Number of symplectic self-dual subspaces}
\label{appendix:symplectic}
In this Appendix, we calculate the size of $\Omega=\{C\subset F_2^{2N}:C=C^{\bot}\}$ and $\Omega_{\mathbf{c}}=\{C\in \Omega:\mathbf{c}\in C\}$, and show that $|\Omega_{\mathbf{c}}|/|\Omega|=1/(2^N+1)$ for any $\mathbf{c}(\neq\mathbf{0})\in F_2^{2N}$, which implies the last equation in Eq. \eqref{eq:symplectic}.

Any symplectic self-dual subspace of $F_2^{2N}$ can be constructed by the following procedure:
\begin{enumerate}
 \item Set $C_0=\{\mathbf{0}\in F_2^{2N}\}$ and $n=0$.
 \item Choose $\mathbf{s}(n+1)\in F_2^{2N}$ so that $\mathbf{s}(n+1)\in C_{n}^{\bot}$ and $\mathbf{s}(n+1)\notin C_{n}$.
 \item Set $C_{n+1}=span\{\mathbf{s}(m)\}_{m=1}^{n+1}$ and increase $n$ by one.
 \item Repeat step 2 to step 3 until no $\mathbf{s}\in F_2^{2N}$ satisfies the condition in step 2.
\end{enumerate}
Since $\{\mathbf{s}(n)\}_n$ is linearly independent, $\dim C_n=n$. Since $\mathbf{s}(k)\odot\mathbf{s}(l)=0$ for any $k$ and $l$, $C_n\subseteq C_n^{\bot}$. Thus, using the procedure, we can obtain symplectic self-dual subspace $C_N$ and its basis $\{\mathbf{s}(n)\}_{n=1}^N$. Conversely, we can easily verify that any symplectic self-dual subspace and any its basis are constructed by the procedure.

In the procedure, we obtain $\prod_{n=0}^{N-1}(2^{2N-n}-2^{n})$ different families of linearly independent vectors $\{\mathbf{s}(n)\}_{n=1}^N$. For any $N$-dimensional subspace $C_N$, there exist $\prod_{n=0}^{N-1}(2^{N}-2^{n})$ different families $\{\mathbf{s}(n)\}_{n=1}^N$ each of which is a basis of $C_N$. Therefore, the number of symplectic self-dual subspaces is given by
\begin{equation}
 |\Omega|=\frac{\prod_{n=0}^{N-1}(2^{2N-n}-2^{n})}{\prod_{n=0}^{N-1}(2^{N}-2^{n})}.
\end{equation}

If we choose $\mathbf{s}(1)=\mathbf{c}(\neq\mathbf{0})$, we obtain $\prod_{n=1}^{N-1}(2^{2N-n}-2^{n})$ different families $\{\mathbf{c},\mathbf{s}(2),\cdots,\mathbf{s}(N)\}$. For any $N$-dimensional subspace $C_N$ containing $\mathbf{c}(\neq\mathbf{0})$, there exist $\prod_{n=1}^{N-1}(2^{N}-2^{n})$ different families $\{\mathbf{c},\mathbf{s}(2),\cdots,\mathbf{s}(N)\}$ each of which is a basis of $C_N$. Therefore, the number of symplectic self-dual subspaces containing $\mathbf{c}(\neq\mathbf{0})$ is given by
\begin{equation}
 |\Omega_{\mathbf{c}}|=\frac{\prod_{n=1}^{N-1}(2^{2N-n}-2^{n})}{\prod_{n=1}^{N-1}(2^{N}-2^{n})}.
\end{equation}

\section{Identification in separable ensembles}
\label{appendix:mixed}
If the mixed state corresponding to each ensemble, $\rho=\sum_{m\in F_2^2}p_m\ket{\Phi_m}\bra{\Phi_m}$, is separable, $2p_m\leq 1$ for all $m\in F_2^2$. Using an equation
\begin{equation}
\label{eq:relativeentropy}
 \sum_{m\in F_2^2}p_m\log (2p_m)=1-H(\mathbf{p}),
\end{equation}
we obtain that if $\rho$ is separable, $H(\mathbf{p})\geq 1$ with equality occurring only when the number of non-zero elements in a probability vector $\mathbf{p}$ is $2$. Using Theorem \ref{theorem:upperbound}, we can verify that if $\rho$ is separable and the success probability of the identification, $\gamma^{(N)}(\mathbf{p})$, is less than $1$ for some finite $N$, $\gamma^{(N)}(\mathbf{p})$ converges to $0$ as $N\rightarrow \infty$.

\nocite{*}
\bibliographystyle{eptcs}
\bibliography{generic}

\end{document}